\documentclass[%
 aip,
 amsmath,amssymb,
 reprint,%
jmp]{revtex4-2}

\usepackage{graphicx}
\usepackage{dcolumn}
\usepackage{bm}
\usepackage[mathlines]{lineno}

\usepackage{amsthm}
\usepackage{amsmath}
\usepackage{amssymb}
\usepackage{mathrsfs}%
\usepackage{mathtools}
\usepackage{mathptmx}
\usepackage{etoolbox}

\makeatletter
\def\@email#1#2{%
 \endgroup
 \patchcmd{\titleblock@produce}
  {\frontmatter@RRAPformat}
  {\frontmatter@RRAPformat{\produce@RRAP{*#1\href{mailto:#2}{#2}}}\frontmatter@RRAPformat}
  {}{}
}%
\makeatother

\newtheorem{theorem}{Theorem}[section]
\newtheorem*{definition}{Definition}
\newtheorem{corollary}{Corollary}[theorem]
\newtheorem{proposition}{Proposition}[theorem]
\newtheorem{lemma}{Lemma}[theorem]

\begin{document}

\preprint{AIP/123-QED}

\title[Functional Measures Associated to Operators]{Functional Measures Associated to Operators}
\author{{Luis A.} {Cede\~no-P\'erez}$^1$ and Hernando Quevedo$^{1,2,3}$}


\affiliation{ Instituto de Ciencias Nucleares, Universidad Nacional Aut\'onoma de M\'exico, AP  70543, Mexico City, Mexico }
\affiliation{ Dipartimento di Fisica and Icra, Universit\`a di Roma “La Sapienza”, Roma, Italy }
\affiliation{ Al-Farabi Kazakh National University, Al-Farabi av. 71, 050040 Almaty, Kazakhstan}

\date{\today}

\begin{abstract}
We show that every operator in $L^{2}$ has an associated measure on a space of functions and prove that it can be used to find  solutions to abstract Cauchy problems, including partial differential equations. We find explicit formulas to compute the integral of functions with respect to this measure and develop approximate formulas in terms of a perturbative expansion. We show that this method can be used to represent solutions of classical equations, such as the diffusion and Fokker-Planck equations, as Wiener and Martin-Siggia-Rose-Jansen-de Dominics integrals, and propose an extension to paths in infinite dimensional spaces.
\end{abstract}


\maketitle



\tableofcontents

\section{Introduction}
\label{sec:int}

Path integration was originally formulated in a rigorous mathematical manner  by Wiener \cite{wiener1923differential,wiener1924average} to study the Brownian motion and phenomena related to diffusion. However, path integration became a common tool of theoretical physics only after Feynman introduced the concept of path integrals in quantum  mechanics as an alternative to the standard procedure of canonical quantization \cite{feynman1948space}. 

From a physical point of view, Feynman's path integral represents the transition amplitude of a particle between two states as an integral over all possible trajectories between these two states, where each trajectory is endowed with an appropriate weight. This simple idea gave rise to one of the preferred tools for performing calculations in quantum mechanics, quantum field theory, and overall in physics, in particular, to study stochastic and irreversible processes; see, for instance, \cite{kleinert2009path,zinn2021quantum,chaichian2018path}.

Further developments of Wiener's path integration include the works by Onsager and Machlup  on irreversible processes \cite{onsager1953fluctuations,machlup1953fluctuations}. Moreover, using the operator formulation proposed by Martin, Siggia, and Rose \cite{martin1973statistical}, Janssen and De Dominics \cite{janssen1976lagrangean,dominicis1976techniques,de1978field} developed an alternative formalism for path integration (known as the MSRJD integral), which  opened the possibility of investigating diverse phenomena associated with non-equilibrium thermodynamics. Moreover, the MSRJD path integration can be understood as a functional representation of classical differential equations, such as the Fokker-Planck and Langevin equations. 

The role of mathematics in these developments has been crucial in order to establish in a rigorous way the notion of functional measures, which are essential for the correct formulation of path integration \cite{kuo2006gaussian,bastianelli2006path}. In particular, it turns out that functional measures can be related to the abstract Cauchy problem for a particular differential operator.  
In this work, we will focus on the study of functional measures by analyzing the characteristics and properties of the  corresponding differential operators.

\subsection{Functional Measures}

A functional measure is a measure in a space of functions. Functional measures play a fundamental role in quantum mechanics and statistical physics, where they were introduced by Feynman in the form of his famous path integral. Feynman's development of functional integration was mathematically non-rigorous but led to very powerful results. In the case of quantum mechanics Feynman's integral is not the integral with respect to a measure, however, it converges in $L^{2}$. The first rigorous non-trivial functional measure was introduced by Wiener in stochastic processes to study Brownian motion. This measure, known as the Wiener measure, can be used to solve partial and stochastic differential equations, particularly, the diffusion equation, which plays a prominent role in statistical physics.

Generalizations of the Wiener measure were explored for its applications in partial and stochastic differential equations, an example being the Wiener-Yeh measure in spaces of functions of two variables 
\cite{yeh1960wiener,bogachev1998gaussian}. However, functional measures were developed the most in the context of analysis and probability, where they were incorporated into the theory of Gaussian measures in Banach spaces. The most notable results of the theory are Prohorov's theorem on the Fourier transform of Gaussian measures and the development of abstract Wiener spaces; see  
\cite{bogachev1998gaussian}
or \cite{kuo2006gaussian}. These results lie at the very foundation of quantum and statistical field theory, where they are used in the field integral formulation of the theories  \cite{glimm2012quantum}. 

The construction of measures is usually a delicate and arduous procedure, and the construction of the Wiener and Wiener-Yeh measure is even more so. It is because of this that there are not many concrete generalizations of the Wiener measure. While Prohorov's theorem does prove the existence of Gaussian measures in Banach spaces, it does this in an abstract way that does not give a concrete form to the measure. In this work, we present a generalization of the Wiener measure that provides a concrete realization of the constructed measure. Furthermore, we do this in a way that preserves the connection with partial differential equations.

\begin{figure}
\includegraphics[scale=0.7]{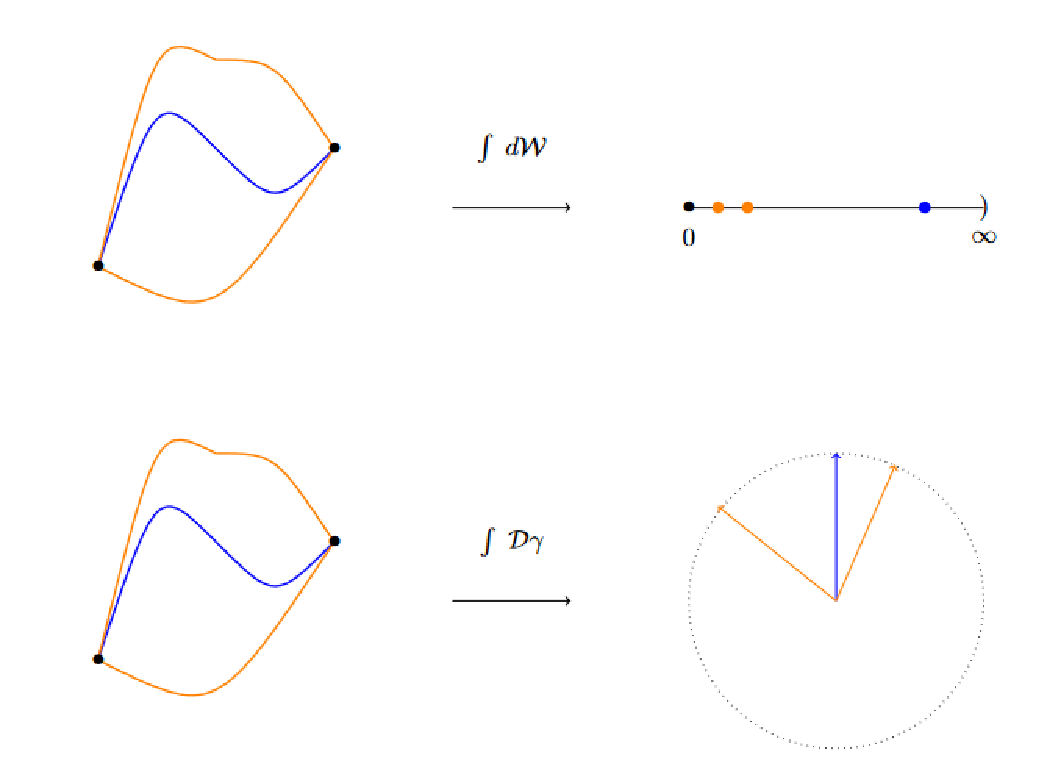}
\caption{Path integral interpretation for the Wiener and Feynman path integrals. In the Wiener integral the path of least action is the one that contributes the most to the average. In the Feynman integral every path has a contribution of the same norm but different phase.}
\end{figure}

\subsection{A Brief Look at the Wiener Measure}

The Wiener measure $\mathcal{W}$ is defined in the space $C(I)$ of continuous functions from a possibly unbounded interval $I$ to an $n$-dimensional Euclidean space, equipped with the supremum norm and its Borel $\sigma$-algebra. The Wiener measure is uniquely defined by its values in cylindrical sets, that is, sets of the form
\begin{equation*}
    \mathcal{C} = \{\gamma\in C(I)\;|\;\gamma(t_{i})\in V_{i},\;i\leq n\},
\end{equation*}
where each $V_{i}$ is a Borel set in $\mathbb{R}^{n}$ and $\{t_{i}\}_{i=1}^{n}$ is a finite set in $I$. The Wiener measure of this set is
\begin{equation*}
    \mathcal{W}(\mathcal{C}) = \int_{V_{1}} \cdots \int_{V_{n}} \prod_{i=1}^{n}\varphi(x_{i}-x_{i-1},t_{i}-t_{i-1})\;dx_{1}\ldots dx_{n},
\end{equation*}
where
\begin{equation}\label{FunciónWiener}
    \varphi(x-y,t) = \frac{1}{\sqrt{4\pi Dt}} \mbox{\Large\(e^{-\frac{|x-y|^{2}}{4Dt}}\)}.
\end{equation}

\begin{figure}[t]
\centering
\includegraphics{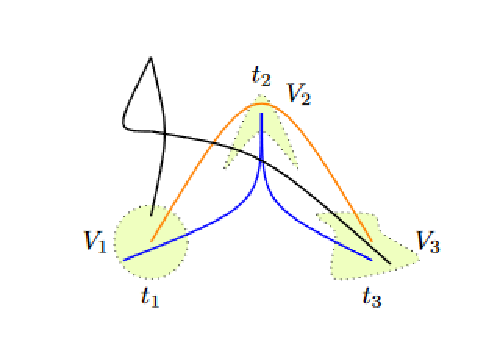}
\caption{Cylindrical set of paths in $\mathbb{R}^{2}$ and sample paths.}
\end{figure}

\begin{figure}[t]
\centering
\includegraphics{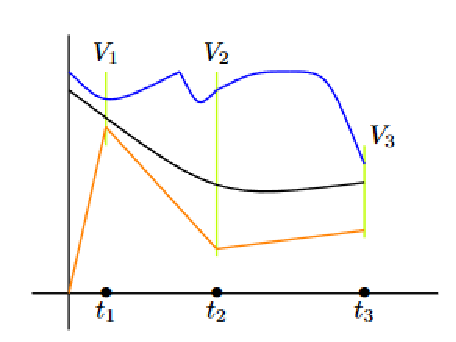}
\caption{Cylindrical set of functions in $\mathbb{R}$, including the parameter space.}
\end{figure}

The classical construction of the Wiener measure consists in verifying that the set function $\mathcal{W}$ is $\sigma$-additive on the class of cylindrical sets and extending it through Caratheodory's theorem \cite{kuo2006gaussian}. This procedure is both complicated and difficult to generalize. An alternative method is to use either the Riesz-Markov theorem or Kolmogorov's extension theorem to construct a measure on a space of rough paths. One then shows that the space of continuous paths has non-zero measure and defines the Wiener measure as the measure of the intersection with this set. To use Kolmogorov's 
extension theorem, one must define a family of measures and verify the compatibility condition. This follows from the following equation, which is satisfied by the function $\varphi$ defined in (\ref{FunciónWiener})
\begin{equation*}
    \int\varphi(x-y,t-s)\varphi(y-z,s-u)\;dy = \varphi(x-z,t-u).
\end{equation*}
This is known as Kolmogorov's compatibility equation, and functions satisfying this equation are known as transition functions.

It is crucial that the Wiener measure is defined in a space of continuous and not rough paths since it allows integrals of functions of paths to be well defined and thus integrated. Of particular importance are expressions of the form
\begin{equation}\label{Integrale-int}
    \int_{C(I)} e^{-\int_{I}V(\gamma(t))\;dt}\;d\mathcal{W}(\gamma),
\end{equation}
since they can be shown to be the solution of the differential equation
\begin{equation*}
    \partial_{t}f = \Delta f + V f.
\end{equation*}

The construction we present here is a generalization of the construction of the Wiener measure with the Riesz-Markov theorem. Both the Riesz-Markov and Kolmogorov's extension theorem can be used for this. The Riesz-Markov theorem has the significant advantage of yielding a measure with a prescribed integral, but Kolmogorov's theorem gives conceptual clarity on what is missing to make possible said generalization.

\subsection{Generalizing the Wiener Measure}

As we previously noted, in order to generalize the construction of the Wiener measure it is necessary that the compatibility condition of Kolmogorov's extension theorem be satisfied. In this case, this condition is simplified to finding transition functions. It is not clear, however, how such functions arise, if at all. This crucial question is answered by noting that the function $\varphi$ defined in Eq.(\ref{FunciónWiener}) satisfies the differential equation
\begin{equation*}
    \partial_{t}\varphi  = \Delta \varphi 
\end{equation*}
with initial condition $\varphi(x-x_{0},0) = \delta_{x_{0}}(x)$. Functions that satisfy this property for an operator $A$ instead of the Laplacian $\Delta$ are called fundamental solutions. This is not coincidental, as we will show that under certain hypotheses on the operator $A$ its fundamental solution must be a transition function, which may take distribution values. 

The relations between operators, fundamental solutions, and transition functions are detailed in Sec. \ref{sec:const}.
In the same section, we proceed to show that every transition function $\varphi$ defines a measure $\mu_{\varphi}$ on a space of rough paths. Since the fundamental solutions of certain operators are transition functions, this shows that there is a correspondence between operators of this certain kind and functional measures. The next question to be answered is how big is the set of continuous paths with respect to $\mu_{\varphi}$. We prove that if $\varphi$ is positive and decays quickly (in a way to be made precise) then the set of continuous paths has full measure with respect to $\mu_{\varphi}$. This will always be true if $\varphi$ is the fundamental solution of an operator of a certain kind.

In Sec. \ref{sec:diff}, we study the relation between the functional measure $\mu_{\varphi}$ associated to an operator $A$ and the abstract Cauchy problem
\begin{equation*}
    \partial_{t}f = A(f) + M_{V} f,
\end{equation*}
where $M_{V}$ is the multiplication operator by the function $V$. In the case of the Wiener measure, this is usually done by using the Lie-Trotter product formula for operator semigroups. This approach has the disadvantage of imposing some conditions on the operator semigroup $\{e^{tA}\}_{t\geq 0}$. We give a more general proof that has the advantage of also developing a perturbative expansion which can be used to approximate integrals of the form (\ref{Integrale-int}).

In Sec. \ref{sec:examples}, we present some examples, namely, 
the diffusion equation, the Fokker-Planck equation, and the n-th derivative operator. We also study the case of a general one-dimensional mechanical Lagrangian  and discuss the possibility of considering higher-order Lagrangians. Notably, we show that this construction generalizes the Wiener measure and the MSRJD integral used in non-equilibrium thermodynamics.

Finally, in Sec. \ref{sec:inf}, we discuss possible generalizations to spaces of paths in infinite dimensional spaces. This requires generalizing the theory of distributions to infinite dimensions. While our discussion is not thorough or complete it brings a possible connection with the theory of LF spaces and abstract harmonic analysis.

\section{Construction of Functional Measures}
\label{sec:const}

\subsection{Transition Functions}

The first step is to define the concepts of fundamental solution and transition function and show that under certain conditions fundamental solutions are also transition functions.

\begin{definition}
Let $A$ be an operator in $L^{2}(\mathbb{R}^{n})$. A \textbf{fundamental solution} for $A$ is a distribution solution $\varphi(x,t)$ to the abstract Cauchy problem
\begin{equation*}
    \partial_{t}\varphi = A(\varphi)
\end{equation*}
with initial condition $\varphi(x - x_{0},0) = \delta_{x_{0}}(x)$.
\end{definition}
This definition is rather imprecise in order to introduce its usual notation. Formally, we require $\varphi(x-y,t)$ to define a family of distributions $\{\varphi_{t}\}_{t\geq 0}$ according to
\begin{equation*}
    \varphi_{t}(x-y) = \varphi(x-y,t),
\end{equation*}
with $\varphi_{0}(x-y) = \delta_{y}(x)$. This implies that the map
\begin{equation*}
    \begin{array}{ccc}
    \mathbb{R}^{+} & \longrightarrow & D'(\mathbb{R}^{n})\\
    t & \longmapsto & \varphi_{t}
    \end{array}
\end{equation*}
is properly defined. Since both its domain and its codomain are topological vector spaces we can require this map to be Gateaux differentiable. We then denote its Gateaux derivative by $\partial_{t}\varphi$, which shows that the abstract Cauchy problem is properly defined. We are also assuming that the operator $A$ has a continuous extension from its domain to $D'(\mathbb{R}^{n})$, which is the case for differential operators.

We now characterize the fundamental solutions of generators of $C_{0}$-semigroups.

\begin{proposition}[Characterization of Fundamental Solutions]
Let $A$ be an operator in $L^{2}(\mathbb{R}^{n})$ that generates a $C_{0}$-semigroup. $A$ admits a fundamental solution if and only if there exists a family of distributions $(\phi(x,t))_{t\in(0,\infty)}$ such that
\begin{equation*}
    e^{tA}f(x) = \int_{\mathbb{R}^{n}} \phi(y-x,t)f(y)\;dy.
\end{equation*}
In any case, the fundamental solution $\varphi$ is given by
\begin{equation*}
    \varphi(x,t) = \begin{cases}
        \phi(x,t) &\textrm{if } t > 0,\\
        \delta_{x} &\textrm{if } t=0.
    \end{cases}
\end{equation*}
\end{proposition}

The central statement of the proposition is that if $A$ admits a fundamental solution then its generated semigroup is basically determined by the fundamental solution. Note that an operator can admit a fundamental solution and not generate a $C_{0}$-semigroup, which implies that both  conditions need to be verified independently.

\begin{proof}
If $A$ admits a fundamental solution $\varphi(x,t)$ we define
\begin{equation*}
    e_{f}(x,t) = \int_{\mathbb{R}^{n}} \varphi(y-x,t)f(y)\;dy.
\end{equation*}
Then
\begin{align*}
    A(e_{f}(x,t)) &= A\int_{\mathbb{R}^{n}} \varphi(y-x,t)f(y)\;dy\\
    &= \int_{\mathbb{R}^{n}} A(\varphi(y-x,t))f(y)\;dy\\
    &= \int_{\mathbb{R}^{n}} \partial_{t}\varphi(y-x,t)f(y)\;dy\\
    &= \partial_{t}\int_{\mathbb{R}^{n}} \varphi(y-x,t)f(y)\;dy\\
    &= \partial_{t}e_{f}(x,t),
\end{align*}
where we assumed that $A$ acts on the variable $x$. Since
\begin{align*}
    e_{f}(x,0) &= \int_{\mathbb{R}^{n}} \varphi(y-x,0)f(y)\;dy\\
    &= \int_{\mathbb{R}^{n}} \delta_{x}(y)f(y)\;dy\\
    &= f(x),
\end{align*}
uniqueness of solutions implies
\begin{align*}
    e^{tA}f(x) &= e_{f}(x,t)\\
    &= \int_{\mathbb{R}^{n}} \varphi(y-x,t)f(y)\;dy.
\end{align*}
Note that since $A$ generates a $C_{0}$-semigroup the solution is indeed unique.

Conversely, if
\begin{equation*}
    e^{tA}f(x) = \int_{\mathbb{R}^{n}} \phi(y-x,t)f(y)\;dy,
\end{equation*}
then, substituting $f(x) = \delta_{x_{0}}(x)$ we find that $\phi(x_{0}-x,t)$ is a distribution solution to the differential equation $\partial_{t}\varphi = A\varphi$ such that $\phi(x_{0}-x,0) = \delta_{x_{0}}(x)$, thus $\phi$ is a fundamental solution.
\end{proof}

We now define transition functions.

\begin{definition}
A function $\varphi(x,t)$ is a \textbf{transition function} if it satisfies the compatibility condition
\begin{equation*}
    \int\varphi(x-y,t-s)\varphi(y-z,s-u)\;dy = \varphi(x-z,t-u).
\end{equation*}
\end{definition}

Once  again, this definition is rather informal. In the classical definition one considers a transition function $\{\varphi_{t}\}_{t > 0}$ as a family of functions in $L^{2}(\mathbb{R}^{n})$. However, there is no ambiguity in considering families of distributions.

The next theorem shows how operators and transition functions are related to each other.

\begin{theorem}[Kolmogorov Compatibility]
Let $A$ be an operator in $L^{2}(\mathbb{R}^{n})$ that generates a $C_{0}$-semigroup. If $A$ admits a fundamental solution $\varphi$ then it is a transition function.
\end{theorem}
\begin{proof}
If $\varphi$ is the fundamental solution of $A$ the previous result implies that
\begin{equation*}
    e^{tA}f(x) = \int_{\mathbb{R}^{n}} \varphi(y-x,t)f(y)\;dy.
\end{equation*}
The semigroup property implies that the two quantities
\begin{align*}
    e^{(t+s)A}f(x) = \int_{\mathbb{R}^{n}} \varphi(y-x,t+s)f(y)\;dy
\end{align*}
and
\begin{align*}
    e^{tA}e^{sA}f(x) &= e^{tA}\left(\int_{\mathbb{R}^{n}} \varphi(y-x,s)f(y)\;dy\right)\\
    &= \int_{\mathbb{R}^{n}}\int_{\mathbb{R}^{n}} \varphi(z-x,t)\varphi(y-z,s)f(y)\;dy\;dz
\end{align*}
are equal. Evaluating in $f(x) = \delta_{x_{0}}(x)$ we obtain that
\begin{equation*}
    \int_{\mathbb{R}^{n}} \varphi(x_{0}-z,s)\varphi(z-x,t)\;dz = \varphi(x_{0}-x,t+s),
\end{equation*}
which is the compatibility condition with different variable names.
\end{proof}

\subsection{Functional Measure Associated to an Operator}

We now proceed to construct functional measures starting from a transition function. We do not assume the transition function to be the fundamental solution of the generator of a $C_{0}$-semigroup when it is not necessary.

Let $\hat{\mathbb{R}}^{n} = \mathbb{R}^{n}\cup\{\infty\}$ be the one point compactification of $\mathbb{R}^{n}$. With this topology $\hat{\mathbb{R}}^{n}$ is homeomorphic to $S^{n}$. For $t > 0$ we define
\begin{equation*}
    \Omega_{t} = \prod_{i\in [0,t]} \hat{\mathbb{R}}^{n},
\end{equation*}
in other words, the set of functions $f\colon 
[0,t] \to \hat{\mathbb{R}}^{n}$. We equip this space with the product topology $\tau_{\prod}$, that is to say, the topology of pointwise convergence. It follows from Tychonoff's theorem that $(\Omega,\tau_{\prod})$ is a compact topological space which is also easily verified to be Hausdorff. For fixed $x,y\in\mathbb{R}^{n}$ we define
\begin{equation*}
    \Omega_{x,y;t} = \{\gamma\in\Omega_{t}\;|\;\gamma(0) = x,\;\gamma(t) = y\}.
\end{equation*}
By a simple net argument one can verify that $\Omega_{x,y;t}$ is closed, thus also compact and Hausdorff in the subspace topology, which we will also denote by $\tau_{\prod}$. Let us consider the Borel $\sigma$-algebra $\sigma\left(\tau_{\prod}\right)$ of this space. We will work on the measurable space $\left(\Omega_{x,y;t},\sigma\left(\tau_{\prod}\right)\right)$.

\begin{definition}
We say that $F\colon \Omega_{x,y;t} \to \mathbb{R}$ is a \textbf{cylinder function} if $F$ is continuous and there exists a set of points $T = \{0\leq t_{1} \leq \cdots \leq t_{j}\}\subset [0,t]$ and a function $\Tilde{F} \colon \prod_{i\leq j}\hat{\mathbb{R}}^{n} \to \mathbb{R}$ such that
\begin{equation*}
    F(\gamma) = \tilde{F}(\gamma(t_{1}),\ldots,\gamma(t_{j})).
\end{equation*}
for $\gamma \in \Omega_{x,y;t}$.
\end{definition}

We denote the set of cylinder functions by $C_{cyl}(\Omega_{x,y;t})$. Note that $C_{cyl}(\Omega_{x,y;t})$ is a linear subspace of $C(\Omega_{x,y;t})$.

\begin{theorem}[Existence of Functional Measures]
Let $\varphi$ be a transition function. There exists a unique regular Borel measure $\mu_{\varphi}^{x,y;t}$ on $\Omega_{x,y;t}$ such that
\begin{equation*}
    \int_{\Omega}F(\gamma)\;d\mu_{\varphi}^{x,y;t} = \int_{\mathbb{R}^{n}}\cdots \int_{\mathbb{R}^{n}}\tilde{F}(x_{1},\ldots,x_{j}) \prod_{i=1}^{j+1} \varphi(x_{i}-x_{i-1},t_{i}-t_{i-1}) \;dx_{1}\cdots dx_{j}
\end{equation*}
for every cylinder function $F$, where we set $t_{0} = 0$, $t_{j+1} = t$, $x_{0} = x
$ and $x_{j+1} = y$. The measure $\mu_{\varphi}^{x,y;t}$ is the \textbf{functional measure} associated to $\varphi$, with fixed endpoints $x$ and $y$ at time $t$.
\end{theorem}

\begin{proof}
We define $\Lambda \colon C_{cyl}(\Omega_{x,y;t}) \to \mathbb{R}$ as
\begin{equation*}
    \Lambda(F) = \int_{\mathbb{R}^{n}}\cdots \int_{\mathbb{R}^{n}}\tilde{F}(x_{1},\ldots,x_{j}) \prod_{i=1}^{j+1} \varphi(x_{i}-x_{i-1},t_{i}-t_{i-1})\;dx_{1}\cdots dx_{j}.
\end{equation*}
It follows from Kolmogorov's compatibility equation that $\Lambda$ is well defined. This is because if $\tilde{F}$ defines a function of $j$ variables then in also defines a function on $j+k$ variables that does not depend on the last $k$ variables. In this case, Kolmogorov's compatibility equation assures that the additional $\varphi$'s do not change the value of $\Lambda$.

It is verified that
\begin{equation*}
    \Lambda(1) = \varphi(x-y,t)
\end{equation*}
and
\begin{equation*}
    |\Lambda(F)| \leq |\varphi(x - y,t)|\; |F|.
\end{equation*}
Then, it follows that $|\Lambda| = |\varphi(x-y,t)|$. On the other hand, from  the Hahn-Banach theorem it follows that $\Lambda$ has a continuous extension of the same norm to all of $C(\Omega_{x,y;t})$, which allows us to consider $\Lambda$ as an element of $C(\Omega_{x,y;t})^{\ast}$. It follows from the Riesz-Markov-Kakutani theorem that there exists a unique signed Borel measure $\mu_{\varphi}^{x,y;t}$ such that for $F\in C_{cyl}(\Omega_{x,y;t})$ we have
\begin{align*}
    \int_{\Omega_{x,y;t}} F\;d\mu_{\varphi} &= \Lambda(F)\\
    &= \int_{\mathbb{R}^{n}}\cdots \int_{\mathbb{R}^{n}}\tilde{F}(x_{1},\ldots,x_{j}) \prod_{i=1}^{j+1} \varphi(x_{i}-x_{i-1},t_{i}-t_{i-1})\;dx_{1}\cdots dx_{j},
\end{align*}
which was to be proved.
\end{proof}

Note that the same result can be achieved by means of Kolomogorov's extension theorem. Indeed, for a finite product subspace the functional $\Lambda$ induces a measure in such space. Since $\varphi$ is a transition function, all these measures are compatible and Kolmogorov's extension theorem yields a measure with the desired properties. It is only left to show that integration with respect to this measure coincides with $\Lambda$, which is a simple technicality.

The previous result also applies when the transition function $\varphi$ takes distribution values. Since $\mu_{\varphi}^{x,y;t}(\Omega_{x,y;t}) = \varphi(x-y;t)$, in case that $\varphi(x-y;t)$ is a distribution, this is actually the distribution evaluated in the function of constant value $1$.

We now show under which conditions this measure can be restricted to the subset of continuous paths. The objective of this is to make integrals of functions of paths integrable functions. The proof we present is adapted from appendix B of \cite{nelson1964feynman}.

\begin{theorem}[Full Measure of Continuous Paths]
Let $\varphi$ be a positive transition function. If
\begin{equation}\label{EqFullMeasureCondition}
    \lim_{\delta\to 0}\int_{|x_{1}-x_{2}|\geq\epsilon}\varphi(x_{2} - x_{1},\delta)\varphi(x_{1} - x,s)\;dx_{1} = 0
\end{equation}
for each $s\in [0,t]$, $x_{2}\in\hat{\mathbb{R}}^{n}$ and $\epsilon > 0$ then the set of elements in $\Omega_{x,y;t}$ that are continuous at every point has full measure respect to $\mu_{\varphi}^{x,y;t}$.
\end{theorem}

Note that fundamental solutions of generators of $C_{0}$-semigroups satisfy condition (\ref{EqFullMeasureCondition}) since $\lim_{\delta \to 0}\varphi(x-y,\delta) = \delta_{y}(x)$. For the proof we will use $D$ to denote the set of elements of $\Omega_{x,y;t}$ that are discontinuous at at least one point. The proof is rather technical and consists, essentially, on rewriting $D$ in terms of cylinder sets to compute its measure.

\begin{proof}
We only need to show that $\mu_{\varphi}^{x,y;t}(D) = 0$. For each $\epsilon,\delta > 0$ consider the set $\omega(\epsilon,\delta)$ consisting on the elements $\gamma\in\Omega_{x,y;t}$ for which there exist $r,s\in [0,t]$ such that $|r-s| < \delta$ and $|\gamma(r) - \gamma(s)| \geq \epsilon$. It is then clear that
\begin{equation*}
    D = \bigcup_{\epsilon > 0} \bigcap_{\delta > 0} \omega(\epsilon,\delta).
\end{equation*}
It is easy to verify that the sets $\omega(\epsilon,\delta)$ are decreasing in $\epsilon$ and increasing in $\delta$. Both the previous union and intersection can be made countable by having $\epsilon$ and $\delta$ be reciprocals of naturals. In this case, if $\delta = \frac{1}{k}$ then $\omega(\epsilon,\delta)$ decreases as $k$ increases, and similarly for $\epsilon$. This implies that  it is possible to use the continuity of the measure to conclude that
\begin{equation*}
    \mu_{\varphi}^{x,y;t}(D) = \lim_{\epsilon \to 0}\lim_{\delta\to 0} \mu_{\varphi}^{x,y;t}(\omega(\epsilon,\delta)).
\end{equation*}
It follows that it is enough to show that $\lim_{\delta\to 0} \mu_{\varphi}^{x,y;t}(\omega(\epsilon,\delta)) = 0$ for any reciprocal of a natural $\epsilon$. For any finite and increasingly ordered set $F = \{t_{i}\}_{i=1}^{N} \subset [0,t]$ such that $|t_{i} - t_{i-1}| < \delta$ we define
\begin{equation*}
    \omega(\epsilon;F) = \{\gamma\in\Omega_{x,y;t}\;|\;\exists\; i\leq N\;\textrm{ such that } |\gamma(t_{i}) - \gamma(t_{i-1})| \geq \epsilon\}.
\end{equation*}
Let $\mathcal{C}_{\delta}$ denote the class of all sets of this form. We have that
\begin{equation*}
        \omega(\epsilon,\delta) = \bigcup_{\omega\in\mathcal{C}_{\delta}}\omega.
\end{equation*}
Furthermore, the class $\mathcal{C}_{\delta}$ is closed under finite unions, thus, inner regularity implies that
\begin{equation*}
    \mu_{\varphi}^{x,y;t}(\omega(\epsilon,\delta)) = \sup_{\omega\in\mathcal{C}_{\delta}}\mu_{\varphi}^{x,y;t}(\omega).
\end{equation*}
We will be able to finally conclude if we can show that there is a bound for $\mu_{\varphi}^{x,y;t}(\epsilon;F)$ that does not depend on $F$ and vanishes as $\delta \to 0$. To simplify notation, let
\begin{equation*}
    B = \{(x_{1},x_{2})\in (\hat{\mathbb{R}}^{n})^{2}\;|\;|x_{1} - x_{2}| < \epsilon\},
\end{equation*}
$\varphi_{i,i-1} = \varphi(x_{i} - x_{i-1};t_{i} - t_{i-1})$ and we will also write $i$ instead of $x_{i}$ when we evaluate inside a function. Notice that the complement of $\omega(\epsilon;F)$, is
\begin{equation*}
    \omega^{C}(\epsilon;F) = \{\gamma\in\Omega_{x,y;t}\;|\;|\gamma(t_{i}) - \gamma(t_{i-1})| < \epsilon,\;\forall i\leq N\},
\end{equation*}
which is a cylinder set since
\begin{equation*}
\chi_{\omega^{C}(\epsilon;F)}(\gamma) = \prod_{i=1}^{N+1}\chi_{B}(\gamma(t_{i}),\gamma(t_{i-1})).
\end{equation*}
It follows that
\begin{equation*}
\mu_{\varphi}^{x,y;t}(\omega^{C}(\epsilon;F)) = \int_{(\hat{\mathbb{R}}^{n})^{N}} \prod_{i=1}^{N+1}\chi_{B}(i,i-1)\varphi_{i,i-1}\;dx_{i}.
\end{equation*}
From this we conclude that
\begin{align*}
    \mu_{\varphi}^{x,y;t}(\omega(\epsilon;F)) &= \varphi(x-y;t) - \int_{(\hat{\mathbb{R}}^{n})^{N}} \prod_{i=1}^{N+1}\chi_{B}(i,i-1)\varphi_{i,i-1}\;dx_{i}\\
    &= \int_{(\hat{\mathbb{R}}^{n})^{N}} \prod_{i=1}^{N+1}\varphi_{i,i-1}\;dx_{i} - \int_{(\hat{\mathbb{R}}^{n})^{N}} \prod_{i=1}^{N+1}\chi_{B}(i,i-1)\varphi_{i,i-1}\;dx_{i}\\
    &= \int_{(\hat{\mathbb{R}}^{n})^{N}} \left(1 - \prod_{i=1}^{N+1}\chi_{B}(i,i-1)\right)\prod_{i=1}^{N+1}\varphi_{i,i-1}\;dx_{i}.
\end{align*}
Let us now note that if we write $\chi_{0} = \chi_{B}$ and $\chi_{1} = \chi_{B^{c}}$ then
\begin{align*}
    1 - \prod_{i=1}^{N+1}\chi_{B}(i,i-1) &= \prod_{i=1}^{N+1}(\chi_{B}(i,i-1) + \chi_{B^{c}}(i,i-1)) - \prod_{i=1}^{N+1}\chi_{B}(i,i-1)\\
    &= \sum_{\sigma\in\{0,1\}^{N+1}}\prod_{i=1}^{N+1}\chi_{\sigma(i)}(i,i-1) - \prod_{i=1}^{N+1}\chi_{B}(i,i-1)\\
    &= \sum_{\sigma\in\{0,1\}^{N+1}\setminus\{0\}}\prod_{i=1}^{N+1}\chi_{\sigma(i)}(i,i-1),
\end{align*}
where 
$\{0,1\}^{N+1}\setminus\{0\}$ is the set of $N+1$-tuples taking values in $\{0,1\}$ with at least one non-zero entry. This function will be integrated, thus, inside the integral only, we can rename variables in such a way that each summand contains a factor of the form $\chi_{B^{c}}(1,0)$. This is possible since each summand contains a factor of $\chi_{B^{C}}$ in some variables. It follows that, when integrated, this function is equal to
\begin{equation*}
     \chi_{B^{c}}(1,0) J(0,\ldots,N+1),
\end{equation*}
where $J$ is some function that is lesser or equal than $1$. We can finally conclude that
\begin{align*}
    \mu_{\varphi}^{x,y;t}(\omega(\epsilon;F)) &= \int_{(\hat{\mathbb{R}}^{n})^{N}} \left(1 - \prod_{i=1}^{N+1}\chi_{B}(i,i-1)\right)\prod_{i=1}^{N+1}\varphi_{i,i-1}\;dx_{i}\\
    &= \int_{(\hat{\mathbb{R}}^{n})^{N}} \chi_{B^{C}}(1,0)J(0,\ldots,N+1)\prod_{i=1}^{N+1}\varphi_{i,i-1}\;dx_{i}\\
    &\leq \int_{(\hat{\mathbb{R}}^{n})^{N}} \chi_{B^{C}}(1,0)\prod_{i=1}^{N+1}\varphi_{i,i-1}\;dx_{i}\\
    &= \int_{|x_{1}-x_{2}|\geq\epsilon}\varphi(y - x_{1},t - t_{1})\varphi(x_{1} - x,s)\;dx_{1}.
\end{align*}
It follows that this integral is an upper bound of $\{\mu_{\varphi}^{x,y;t}(\omega)\;|\;\omega\in\mathcal{C}_{\delta}\}$, thus
\begin{equation*}
    \mu_{\varphi}^{x,y;t}(\omega(\epsilon,\delta)) \leq \int_{|x_{1}-x_{2}|\geq\epsilon}\varphi(y - x_{1},t - t_{1})\varphi(x_{1} - x,s)\;dx_{1}.
\end{equation*}
Our hypothesis implies that this integral bound vanishes as $\delta \to 0$ and we conclude the result.

\begin{figure}[t]
\centering
\includegraphics{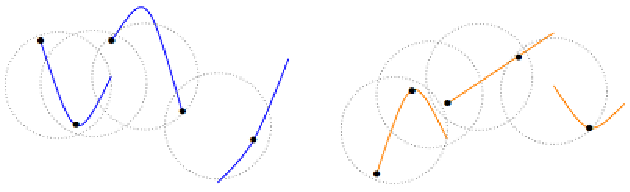}
\caption{Examples of paths in $\omega^{C}(\epsilon;F)$. This is finally a cylinder set as it checks the distance between points at fixed times, eventhough the dotted circles may change for each path.}
\end{figure}

\end{proof}

If $C_{x,y;t}$ denotes the set of continuous paths in $\Omega_{x,y;t}$, the previous theorem implies that for a Borel measurable $E$ the function
\begin{equation*}
    \nu(E) = \mu_{\varphi}^{x,y;t}(C_{x,y;t}\cap E)
\end{equation*}
defines a non-trivial measure in $C_{x,y;t}$. We will denote this measure $\mu_{\varphi}^{x,y;t}$ as before. If $|\cdot|_{\infty}$ denotes the supremum norm of $C_{x,y;t}$ and $\tau_{\infty}$ its topology, then $\tau_{\prod} \subset \tau_{\infty}$. This implies that $\sigma(\tau_{\prod}) \subset \sigma(\tau_{\infty})$. In fact, the equality holds. To see this let $\{r_{n}\}_{n\in\mathbb{N}}$ be a numeration of the rationals in $[0,t]$. Let $V(I_{1},\ldots,I_{j};t_{1},\ldots,t_{j}) = \{f\colon [0,t] \to \hat{\mathbb{R}}^{n}\;|\;f(t_{j})\in I_{j}\}$ be a standard basis of $\tau_{\prod}$. Continuity implies that
\begin{align*}
    B_{1}^{|\cdot|_{\infty}}(0) = \bigcap_{n\in\mathbb{N}}V((-1,1);r_{n}).
\end{align*}
Since $C([0,t])$ is separable every open set is the countable union of open balls. Restricting to $C_{x,y;t}$ we obtain $\sigma(\tau_{\infty}) \subset \sigma(\tau_{\prod})$.

A few comments are in place. First, the space $(C_{x,y;t},\tau_{\infty})$ is no longer compact, however, it is a Banach space. Thus $\mu_{\varphi}^{x,y;t}$ is a measure in a Banach space. Second, the measure $\mu_{\varphi}^{x,y;t}$ restricts to $C_{x,y;t}$ in a non-trivial manner only if $\varphi$ is non negative or, more generally, if one can show that $\mu_{\varphi}^{x,y;t}(C_{x,y;t}) > 0$. This does not contradict other results on measures on Banach spaces, particularly, the complex version of Kolmogorov's extension theorem and the non-existence of Feynman's measure.

Since $C_{x,y;t}$ consists only of continuous paths, maps of the form
\begin{equation*}
\begin{array}{cccc}
     \mathcal{F} \colon & \Omega_{x,y;t} & \longrightarrow & \mathbb{R}\\
     & \gamma & \longmapsto & f\left(\int_{0}^{t}V(\gamma(s))\;ds\right)
\end{array}
\end{equation*}
are defined in $C_{x,y;t}$. Our next result shows they are in fact measurable and integrable under mild conditions.

\begin{theorem}[Functional Integral Formula]
Let $\varphi$ be a transition function such that continuous functions have non-zero measure respect to its functional measure $\mu_{\varphi}^{x,y;t}$. If $V\in C(\mathbb{R}^{n})$ and $f\in C(\mathbb{R})$ are such that the function
\begin{equation*}
\begin{array}{cccc}
     \mathcal{F} \colon & \Omega_{x,y;t} & \longrightarrow & \mathbb{R}\\
     & \gamma & \longmapsto & f\left(\int_{0}^{t}V(\gamma(s))\;ds\right)
\end{array}
\end{equation*}
is bounded then $\mathcal{F}$ is integrable respect to $\mu_{\varphi}^{x,y;t}$ and
\begin{align*}
    \int_{\Omega_{x,y;t}} f\left(\int_{0}^{t}V(\gamma(s))\;ds\right)\;d\mu_{\varphi}^{x,y;t}(\gamma) = &\lim_{N\to\infty} \int \cdots\int f\left(\sum_{k=1}^{N}V(x_{k})\Delta t(N)\right)\\
    &\prod_{i=1}^{N+1} \varphi(x_{i}-x_{i-1},t_{i}-t_{i-1})\;dx_{1}\cdots dx_{N},
\end{align*}
where we set $t_{0} = 0$, $t_{N+1} = t$, $x_{0} = x$ and $x_{N+1} = y$
\end{theorem}
\begin{proof}
We define $T_{N} = \{t_{k} = \frac{tk}{N}\;|\;k\leq N\}$ and $\Delta t(N) = \frac{t}{N}$. Since $V(\gamma(s))$ is integrable for every $\gamma\in \Omega_{x,y;t}$ we have that
\begin{equation*}
    \int_{0}^{t}V(\gamma(s))\;ds = \lim_{N\to\infty} \sum_{k=1}^{N}V(\gamma(t_{k}))\Delta t(N).
\end{equation*}
Each of the functions $V_{k}\colon \Omega \to \mathbb{R}$ given by
\begin{equation*}
    V_{k}(\gamma) = V(\gamma(t_{k}))
\end{equation*}
is continuous which implies they are Borel measurable. This implies that the functions
\begin{equation*}
    F_{N}(\gamma) = \sum_{k=1}^{N}V(\gamma(t_{k}))\Delta t(N)
\end{equation*}
are Borel measurable and, since the pointwise limit of Borel measurable functions is Borel measurable, the function
\begin{equation*}
    \gamma \longmapsto \int_{0}^{t}V(\gamma(s))\;ds
\end{equation*}
is Borel measurable and $\mathcal{F}$ is also Borel measurable. Since $\mathcal{F}$ is bounded by hypothesis and $\mu_{\varphi}^{x,y;t}$ is a finite measure the dominated convergence theorem implies that $\mathcal{F}$ is integrable and
\begin{align*}
    \int_{\Omega_{x,y;t}} f\left(\int_{0}^{t}V(\gamma(s))\;ds\right)\;d\mu_{\varphi}^{x,y;t}(\gamma) &= \int_{\Omega_{x,y;t}} \lim_{N\to\infty}f\left(\sum_{k=1}^{N}V(\gamma(t_{k}))\Delta t(N)\right)\;d\mu_{\varphi}^{x,y;t}(\gamma)\\
    &= \lim_{N\to\infty}\int_{\Omega_{x,y;t}} f\left(\sum_{k=1}^{N}V(\gamma(t_{k}))\Delta t(N)\right)\;d\mu_{\varphi}^{x,y;t}(\gamma)\\
    &= \lim_{N\to\infty}\int_{\Omega_{x,y;t}} f\left(F_{N}(\gamma)\right)\;d\mu_{\varphi}^{x,y;t}(\gamma)
\end{align*}
Since the functions $f\left(F_{N}(\gamma)\right)$ are cylinder functions we have that
\begin{align*}
    \int_{\Omega_{x,y;t}} f\left(\int_{0}^{t}V(\gamma(s))\;ds\right)\;d\mu_{\varphi}^{x,y;t}(\gamma) = &\lim_{N\to\infty}\int_{\Omega_{x,y;t}} f\left(F_{N}(\gamma)\right)\;d\mu_{\varphi}^{x,y;t}(\gamma)\\
    = &\lim_{N\to\infty} \int \cdots\int f\left(\sum_{k=1}^{N}V(x_{k})\Delta t(N)\right)\\
    &\prod_{i=1}^{N+1} \varphi(x_{i}-x_{i-1},t_{i}-t_{i-1})\;dx_{1}\cdots dx_{N}
\end{align*}
which concludes the proof of the theorem.
\end{proof}

To apply the preceding theorem one must verify that the function
\begin{equation*}
    \gamma \longmapsto f\left(\int_{0}^{t}V(\gamma(s))\;ds\right)
\end{equation*}
is bounded. This is the case in the following situations:
\begin{enumerate}
    \item $f$ is bounded,
    \item $f$ is increasing and $V$ is bounded from above,
    \item $f$ is decreasing and $V$ is bounded from below.
\end{enumerate}

\section{Functional Measures and Differential Equations}
\label{sec:diff}

\subsection{Fundamental Solution as a Functional Integral}

We now confront the problem of whether the measure $\mu_{\varphi}^{x,y;t}$ can be used to solve differential equations related to an operator $A$ when $\varphi$ is the fundamental solution of $A$. The answer is affirmative under mild conditions, however, the proof  requires some rather tiresome calculations. To motivate the results that follow we first prove a particular case which covers a large amount of operators. In order for the functional measure $\mu_{\varphi}^{x,y;t}$ to have a non-trivial restriction to $C_{x,y;t}$ we assume it to be positive. We will remove this restriction in the following subsection.

\begin{theorem}[Funtional Integral Solution to a PDE. Particular case.]
Let $A$ be an operator in $L^{2}$ that generates a $C_{0}$-semigroup for which there exists a constant $\gamma$ such that
\begin{equation*}
    |e^{tA}|_{B(L^{2})}\leq e^{\gamma t}.
\end{equation*}
If aditionally $A$ admits a fundamental solution $\varphi$ and $V$ is bounded from below then the function
\begin{equation*}
    f(x,t) = \int_{C_{x,y;t}}\mbox{\Large\(e^{-\int_{0}^{t}V(\gamma(s))\;ds} \)} f_{0}(\gamma(s))\;d\mu_{\varphi}^{x,y,t}(\gamma)
\end{equation*}
is a solution to the abstract Cauchy problem
\begin{equation*}
    \partial_{t}f = A(f) + M_{V} (f)
\end{equation*}
with initial condition $f(x,0) = f_{0}(x)$. In particular, the operator $A + M_{V}$ admits a fundamental solution given by
\begin{equation*}
    \varphi_{V}(x-y,t) = \int_{C_{x,y;t}}\mbox{\Large\(e^{-\int_{0}^{t}V(\gamma(s))\;ds} \)}\;d\mu_{\varphi}^{x,y,t}(\gamma).
\end{equation*}
\end{theorem}

Note that this applies, in particular, if $A$ is an elliptic differential operator, provided it admits a fundamental solution \cite{evans2022partial,renardy2004introduction}. This is the case, for example, for any elliptic differential operator with constant coefficients  \cite{rudin1991functional}.

\begin{proof}
Since the exponential is an increasing function and we are assuming that $V$ is bounded from below, the functional integral formula theorem implies that the integral in the conclusion exists. Since $M_{V}$ is a product operator we have
\begin{equation*}
    (e^{tM_{V}}f)(x) = e^{V(x)}f(x),
\end{equation*}
furthermore,
\begin{align*}
    \nonumber|e^{tM_{V}}|_{B(L^{2})} &\leq \sup_{\vec{x}\in\mathbb{R}^{n}}|e^{-t V(\vec{x})}|\\
    \nonumber&= \mbox{\large\(e^{-t \inf_{\vec{x}\in\mathbb{R}^{n}}V(\vec{x})}\)}\\
    &= \mbox{\large\(e^{t \alpha}\)}
\end{align*}
with $\alpha = -\inf_{\vec{x}\in\mathbb{R}^{n}}V(\vec{x})$. This implies that
\begin{equation*}
    \left|\left(\mbox{\large\( e^{\frac{t}{n}A}e^{\frac{t}{n}M_{V}}\)}\right)^{n}\right|_{B(L^{2})} \leq e^{(\alpha + \gamma)t}
\end{equation*}
for every $n\in\mathbb{N}$ and $t\geq 0$. Thus the hypothesis for the Lie-Trotter product formula for $C_{0}$-semigroups are satisfied implying that
\begin{equation*}
    \mbox{\large\(e^{t(A+B)}\)} = \lim_{n\to\infty}\left(\mbox{\large\( e^{\frac{t}{n}A}e^{\frac{t}{n}B}\)}\right)^{n}
\end{equation*}
in the strong operator topology of $B(L^{2})$; see \cite{engel2000one}. Since
\begin{equation*}
    (e^{tM_{V}}f)(x) = e^{V(x)}f(x)
\end{equation*}
and
\begin{equation*}
    e^{tA}f(x) = \int_{\mathbb{R}^{n}} \varphi(y-x,t)f(y)\;dy,
\end{equation*}
we obtain the desired conclusion by substituting into the Lie-Trotter formula and applying the functional integral formula.
\end{proof}

The remaining results of this section will allow us to generalize the previous theorem. Even if one cares only about elliptic operators, the following results are of interest since they provide approximate expressions for the functional integral formula.

For the remainder of the section we will assume that $V$ is a continuous function bounded from below and define
\begin{equation*}
    \varphi_{V}(x-y,t) = \int_{C_{x,y;t}}\mbox{\Large\(e^{-\int_{0}^{t}V(\gamma(s))\;ds} \)}\;d\mu_{\varphi}^{x,y,t}(\gamma).
\end{equation*}

\begin{theorem}[Perturbative Expansion]
If $A$ is an operator in $L^{2}$ that generates a $C_{0}$-semigroup and admits a fundamental solution $\varphi$
then
\begin{equation*}
    \varphi_{V}(x-y,t) = \sum_{n=0}^{\infty}\frac{(-1)^{n}}{n!}\varphi_{V}^{n}(x-y,t),
\end{equation*}
where
\begin{align*}
    \varphi_{V}^{n}(x-y,t) = &\int_{0}^{t}\cdots\int_{0}^{t} \int \cdots \int \varphi(x-x_{n},t-s_{n})\\
    &\prod_{i=1}^{n}V(x_{i})\varphi(x_{i} - x_{i-1},s_{i}-s_{i-1})\;dx_{1}\ldots dx_{n}\;ds_{1}\ldots ds_{n}.
\end{align*}
\end{theorem}
\begin{proof}
Our hypothesis implies that the integrand of $\varphi_{V}$ is a bounded function on the variable $\gamma$. Since $\mu_{\varphi}^{x,y;t}$ is a finite measure in this measurable space it follows from the dominated convergence theorem that
\begin{align*}
    \varphi_{V}(x-y,t) &= \int_{C_{x,y;t}}\mbox{\Large\(e^{-\int_{0}^{t}V(\gamma(s))\;ds} \)}\;d\mu_{\varphi}^{x,y;t}(\gamma)\\
    &= \int_{C_{x,y;t}}\sum_{n=0}^{\infty}\frac{1}{n!}\left(-\int_{0}^{t}V(\gamma(s))\;ds\right)^{n}\;d\mu_{\varphi}^{x,y;t}(\gamma)\\
    &= \sum_{n=0}^{\infty}\frac{(-1)^{n}}{n!}\int_{C_{x,y;t}}\left(\int_{0}^{t}V(\gamma(s))\;ds\right)^{n}\;d\mu_{\varphi}^{x,y;t}(\gamma).
\end{align*}
The result will follow once we prove that
\begin{equation*}
    \int_{C_{x,y;t}}\left(\int_{0}^{t}V(\gamma(s))\;ds\right)^{n}\;d\mu_{\varphi}^{x,y;t}(\gamma) = \varphi_{V}^{n}(x-y,t).
\end{equation*}
We observe that
\begin{align}\label{Eqintegraln}
    \notag\left(\int_{0}^{t}V(\gamma(s))\;ds\right)^{n} &= \prod_{i=1}^{n}\int_{0}^{t}V(\gamma(s_{i}))\;ds_{i}\\
    &= \int_{0}^{t}\cdots\int_{0}^{t}\prod_{i=1}^{n} V(\gamma(s_{i}))\;ds_{1}\ldots ds_{n}.
\end{align}
From Fubini's theorem and the fact that $\gamma \longmapsto \prod_{i=1}^{n} V(\gamma(s_{i}))$ is a cylinder function we obtain
\begin{align*}
    \int_{C_{x,y;t}}\left(\int_{0}^{t}V(\gamma(s))\;ds\right)^{n}\;d\mu_{\varphi}^{x,y;t}(\gamma) = &\int_{C_{x,y;t}}\int_{0}^{t}\cdots\int_{0}^{t}\prod_{i=1}^{n} V(\gamma(s_{i}))\;ds_{1}\ldots ds_{n}\;d\mu_{\varphi}^{x,y;t}(\gamma)\\
    = &\int_{0}^{t}\cdots\int_{0}^{t}\int_{C_{x,y;t}}\prod_{i=1}^{n} V(\gamma(s_{i}))\;d\mu_{\varphi}^{x,y;t}(\gamma)\;ds_{1}\ldots ds_{n}\\
    = &\int_{0}^{t} \int \varphi(x-x_{n},t-s_{n})\\
    &\prod_{i=1}^{n}V(x_{i})\varphi(x_{i} - x_{i-1},s_{i}-s_{i-1})\;dx\;ds,
\end{align*}
from which we conclude the desired result.
\end{proof}

The previous perturbative expansion has the advantage that each term is easy to compute. For our theoretical considerations, however, the following corollary will be more useful.

\begin{corollary}
If $A$ is an operator in $L^{2}$ that generates a $C_{0}$-semigroup and admits a fundamental solution $\varphi$,
then
\begin{equation*}
    \varphi_{V}(x-y,t) = \sum_{n=0}^{\infty} \frac{(-1)^{n}}{n!}K_{V}^{n}(x-y,t),
\end{equation*}
where
\begin{align*}
K_{V}^{n}(x-y,t) =
n!&\int_{0}^{t}\int_{0}^{t_{1}}\cdots\int_{0}^{t_{n-1}} \int\cdots \int \varphi(x-x_{n},t-s_{n})\\
&\prod_{i=1}^{n}V(x_{i})\varphi(x_{i} - x_{i-1},s_{i}-s_{i-1})\;dx_{1}\ldots dx_{n}\;dt_{1}\ldots dt_{n}.
\end{align*}
\end{corollary}

\begin{proof}
We follow the same procedure for the proof of the previous theorem, but instead of using equation (\ref{Eqintegraln}) we use the following formula
\begin{equation*}
    \left(\int_{0}^{t}V(\gamma(s))\;ds\right)^{n} = n! \int_{0}^{t}\int_{0}^{t_{1}}\cdots\int_{0}^{t_{n-1}}\prod_{i=1}^{n}V(\gamma(s_{i}))\;dt_{1}\ldots d_{t_{n}},
\end{equation*}
which is easily proven by induction. 
\end{proof}

From this perturbative expansion we obtain that the functions $K_{V}^{n}$ satisfy a recurrence relation. This will ultimately lead to an integral equation and the original partial differential equation.

\begin{lemma}[Recurrence Relation]
Let $A$ be an operator in $L^{2}$ that generates a $C_{0}$-semigroup and admits a fundamental solution $\varphi$. The previously defined functions $K_{V}^{n}$ satisfy the recurrence relation
\begin{equation*}
    K_{V}^{n+1}(x-y,t) = (n+1)\int_{0}^{t}\int \varphi(x-z,t-s)V(z) K_{V}^{n}(z-y,s)\;dz\;ds.
\end{equation*}
\end{lemma}
\begin{proof}
From the definition of the $K_{V}^{n}$ we have that
\begin{align*}
    K_{V}^{n+1}(x-y,t) = &(n+1)!\int_{0}^{t}\int_{0}^{t_{1}}\cdots\int_{0}^{t_{n}} \int\cdots \int \varphi(x-x_{n+1},t-s_{n+1})\\
    &\;\prod_{i=1}^{n+1}V(x_{i})\varphi(x_{i} - x_{i-1},s_{i}-s_{i-1})\;dx_{1}\ldots dx_{n+1}\;dt_{1}\ldots dt_{n+1}\\
    = &(n+1)\int_{0}^{t} \int \varphi(x-x_{n+1},t-s_{n+1}) V(x_{n+1})\\
    &\bigg(n! \int_{0}^{t_{1}}\cdots\int_{0}^{t_{n}} \int\cdots \int \varphi(x_{n+1} - x_{n},s_{n+1}-s_{n})\\
    &\prod_{i=1}^{n}V(x_{i})\varphi(x_{i} - x_{i-1},s_{i}-s_{i-1})\;dx_{1}\ldots dx_{n}\;dt_{1}\ldots dt_{n}\bigg)dt_{n+1}\;dx_{n+1}\\
    = &(n+1)\int_{0}^{t} \int \varphi(x-x_{n+1},t-t_{n+1}) V(x_{n+1})K_{V}^{n}(x_{n+1}-y,t_{n+1})dt_{n+1}\;dx_{n+1}.
\end{align*}
Renaming variables we obtain the result.
\end{proof}

\begin{theorem}[Integral Equation]
Let $A$ be an operator in $L^{2}$ that generates a $C_{0}$-semigroup and admits a fundamental solution $\varphi$. The function $\varphi_{V}$ satisfies the integral equation
\begin{equation*}
    \varphi_{V}(x-y,t) = \varphi(x-y,t) - \int_{0}^{t}\int \varphi(x-z,t-s)V(z) \varphi_{V}(z,s)\;dz\;ds.
\end{equation*}
\end{theorem}
\begin{proof}
The perturbative expansion and the recurrence relation imply that
\begin{align*}
    \varphi_{V}(x-y,t) &= \sum_{n=0}^{\infty} \frac{(-1)^{n}}{n!}K_{V}^{n}(x-y,t)\\
    &= \varphi(x-y,t) + \sum_{n=1}^{\infty} \frac{(-1)^{n}}{n!}K_{V}^{n}(x-y,t)\\
    &= \varphi(x-y,t) + \sum_{n=0}^{\infty} \frac{(-1)^{n+1}}{(n+1)!}K_{V}^{n+1}(x-y;t)\\
    &= \varphi(x-y,t) + \sum_{n=0}^{\infty} \frac{(-1)^{n+1}}{(n+1)!}(n+1)\int_{0}^{t}\int \varphi(x-z,t-s)V(z) K_{V}^{n}(z-y,s)\;dz\;ds\\
    &= \varphi(x-y,t) - \int_{0}^{t}\int  \varphi(x-z,t-s)V(z)\left(\sum_{n=0}^{\infty}\frac{(-1)^{n}}{n!} K_{V}^{n}(z-y,s)\right)\;dz\;ds\\
    &= \varphi(x-y,t) - \int_{0}^{t}\int \varphi(x-z,t-s)V(z) \varphi_{V}(z,s)\;dz\;ds.
\end{align*}
\end{proof}

We finally obtain the main result of the section.

\begin{theorem}[Funtional Integral Solution to a PDE. General case.]
Let $A$ be an operator in $L^{2}$ that generates a $C_{0}$-semigroup and admits a fundamental solution $\varphi$. Then the function
\begin{equation*}
    f(x,t) = \int_{C_{x,y;t}}\mbox{\Large\(e^{-\int_{0}^{t}V(\gamma(s))\;ds} \)} f_{0}(\gamma(s))\;d\mu_{\varphi}^{x,y;t}(\gamma)
\end{equation*}
is a solution to the abstract Cauchy problem
\begin{equation*}
    \partial_{t}f = A(f) - M_{V} (f)
\end{equation*}
with initial condition $f(x,0) = f_{0}(x)$. In particular, the operator $A + M_{V}$ admits a fundamental solution given by
\begin{equation*}
    \varphi_{V}(x-y,t) = \int_{C_{x,y;t}}\mbox{\Large\(e^{-\int_{0}^{t}V(\gamma(s))\;ds} \)}\;d\mu_{\varphi}^{x,y;t}(\gamma).
\end{equation*}
\end{theorem}
\begin{proof}
Differentiating the integral equation we obtain
\begin{align*}
    \partial_{t}\varphi_{V}(x-y,t) &= \partial_{t}\varphi(x-y,t) - \partial_{t}\int_{0}^{t}\int \varphi(x-z,t-s)V(z) \varphi_{V}(z,s)\;dz\;ds\\
    &= A\varphi(x-y,t) - \int \varphi(x-z,0)V(z) \varphi_{V}(z,t)\;dz\\
    &= A\varphi(x-y,t) - \int \delta_{x}(z) V(z) \varphi_{V}(z,t)\;dz\\
    &= A\varphi(x-y,t) - V(x) \varphi_{V}(x,t).
\end{align*}
To conclude we need to show that $A\varphi = A\varphi_{V}$. Once again, the integral equation implies that $\varphi_{V}$ is of the form
\begin{equation*}
    \varphi_{V}(x-y,t) = \varphi(x-y,t) + f(x,t),
\end{equation*}
where $f$ is the second term of the integral equation. This term notably has no dependence on $y$. Since the dependence of $\varphi_{V}$ on $x$ and $y$ is given in its first , we must have that any operator $A$ that acts on $x$ or $y$ must act on the same way on the other one, except, perhaps, for a minus sign in each of its summands. For simplicity we assume that $A$ has only one summand, which implies that 
\begin{equation*}
    A_{x}\varphi_{V} = \pm A_{y}\varphi_{V}.
\end{equation*}
From this it follows that
\begin{align*}
    A_{x}\varphi_{V} &= \pm A_{y}\varphi_{V}\\
    &= \pm A_{y}\varphi + A_{y}f(x,t)\\
    &= \pm A_{y}\varphi\\
    &= A_{x}\varphi.
\end{align*}
\end{proof}

\subsection{Complex Transition Functions}

Suppose now that a complex transition function $\varphi(x-y,t)$ is given. Since its functional measure $\mu_{\varphi}^{x,y;t}$ may not have a non-trivial restriction to $C_{x,y;t}$ the expression
\begin{equation*}
    \varphi_{V}(x-y,t) = \int_{\Omega_{x,y;t}}\mbox{\Large\(e^{-\int_{0}^{t}V(\gamma(s))\;ds} \)}\;d\mu_{\varphi}^{x,y;t}(\gamma)
\end{equation*}
is ill-defined. However, if $\varphi$ is the fundamental solution of an operator $A$ that satisfies the hypotheses of the Lie-Trotter product formula,  the expression
\begin{equation*}
    \lim_{N\to\infty}\int_{\mathbb{R}^{n}}\cdots \int_{\mathbb{R}^{n}}\mbox{\Large\(e^{\sum_{i=1}^{N}V(x_{i})\Delta t(N)}\)} \prod_{i=1}^{N+1} \varphi(x_{i}-x_{i-1},t_{i}-t_{i-1}) \;dx_{1}\cdots dx_{N}
\end{equation*}
is a fundamental solution for the abstract Cauchy problem
\begin{equation}\label{EqCauchyPerturbada}
    \partial_{t}f = A(f) - M_{V}(f).
\end{equation}
We introduce a sequence of cylinder functions $e_{V,n}\colon \Omega_{x,y;t} \to \mathbb{R}$ given by
\begin{equation*}
    e_{V,n}(\gamma) = \mbox{\Large\(e^{\sum_{k=1}^{n}V(\gamma(t_{k}))\Delta t(n)}\)},
\end{equation*}
where $t_{k} = k\frac{t}{n}$ and $\Delta t(n) = \frac{t}{n}$. We get that the fundamental solution of (\ref{EqCauchyPerturbada}) can be written in terms of $\mu_{\varphi}^{x,y;t}$ as
\begin{equation*}
    \varphi_{V}(x-y,t) = \lim_{n\to\infty} \int_{\Omega_{x,y;t}} e_{V,n}(\gamma)\;d\mu_{\varphi}^{x,y;t}(\gamma).
\end{equation*}
Note however, that the Lie-Trotter formula does not provide pointwise convergence, but convergence in $L^{2}$. In fact, this limit may not exist for different values of $x$, $y$ and $t$. However, it does exist except on a set of measure zero and defines an element of $L^{2}$.

This method for describing the fundamental solution as limits of integrals with respect to $\mu_{\varphi}^{x,y;t}$ is related to Donsker's flat integral used to deal with the Wiener measure \cite{kuo2006gaussian}.

This can be further generalized. If $\varphi$ is the fundamental solution of an operator $A$ and one can show that the limit
\begin{equation*}
    \lim_{n\to\infty} \int_{\Omega_{x,y;t}} e_{V,n}(\gamma)\;d\mu_{\varphi}^{x,y;t}(\gamma)
\end{equation*}
exists except on a set of measure zero as a distribution, then all the results of the previous section follow in the same manner. Particularly, the perturbative expansion and the integral and partial differential equations hold.

\subsection{Interpretation of Results}

Suppose that a particle in $\hat{\mathbb{R}}^{n}$ follows a natural motion which can not be described exactly. However, a probabilistic description can be achieved in the sense that the probability that the particle is at the point $x$ and after a time $t$ it is at the point $y$ is given by $\varphi(x-y,t)$. In this context, the compatibility equation
\begin{equation*}
    \int\varphi(x-y,t-s)\varphi(y-z,s-u)\;dy = \varphi(x-z,t-u)
\end{equation*}
simply represents that the position of the particle does not depend on previous positions at previous times, that is, the particle has "no memory" of its position.  Let a certain quantity that can be measured be defined on the space and be described by a function $f\colon \hat{\mathbb{R}}^{n} \to \mathbb{R}$. If the position of the particle could be known exactly as $\gamma(t)$ then the value of this quantity would be given as a function of time by $f(\gamma(t))$. If the particle moves from $x$ to $y$ in a time $t$ then the average position of the particle in this period of time is
\begin{equation*}
    \frac{1}{t}\int_{0}^{t}\gamma(s)\;ds
\end{equation*}
and the average value of the quantity $f$ is given by
\begin{equation*}
    \frac{1}{t}\int_{0}^{t}f(\gamma(s))\;ds,
\end{equation*}
which can be approximated, rather roughly, by
\begin{equation*}
    f\left(\frac{1}{t}\int_{0}^{t}\gamma(s)\;ds\right).
\end{equation*}
Since we are assuming that the position $\gamma(t)$ can not be known exactly, the average value of the quantity $f$ in this period of time is
\begin{equation*}
    \int f\left(\frac{1}{t}\int_{0}^{t}\gamma(s)\;ds\right) \;d\mu_{\varphi}^{x,y;t}(\gamma).
\end{equation*}
Thus, the knowledge of the probability $\varphi(x-y,t)$ is enough to know the value of different quantities associated to the particle.

Let us now suppose that the natural motion of the particle is modified by the introduction of a potential $V(x)$. Subject to this modification, the probability that the particle moves from $x$ to $y$ in a time $t$ is now given by
\begin{equation*}
    \varphi_{V}(x-y,t) = \int_{C_{x,y;t}}\mbox{\Large\(e^{-\int_{0}^{t}V(\gamma(s))\;ds} \)}\;d\mu_{\varphi}^{x,y;t}(\gamma),
\end{equation*}
except by constant given by normalization. Thus $\varphi_{V}(x-y,t)$ represents the new probability for the motion of the particle. Furthermore, the average value of any measureable quantity $f$ is given by
\begin{equation*}
    \int_{C_{x,y;t}}f(\gamma(t))\;\mbox{\Large\(e^{-\int_{0}^{t}V(\gamma(s))\;ds} \)}\;d\mu_{\varphi}^{x,y;t}(\gamma).
\end{equation*}

According to the perturbative expansion, we have that
\begin{equation*}
    \varphi_{V}(x-y,t) = \sum_{n=0}^{\infty}\frac{(-1)^{n}}{n!}\varphi_{V}^{n}(x-y,t).
\end{equation*}
Thus, the probability $\varphi_{V}(x-y,t)$ of the particle moving from $x$ to $y$ in time $t$ is given by summing the contributions $\varphi_{V}^{n}(x-y,t)$. From this, the interpretation of the $n$-th summand $\varphi_{V}^{n}(x-y,t)$ is that it represents the probability of a particle moving under the natural motion of the space and interacting $n$-times with the potential in the period of time $t$, that is, the particle is scattered $n$ times in this period of time.

\begin{figure}[t]
\centering
\includegraphics{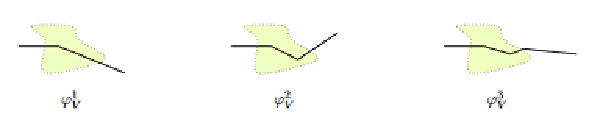}
\caption{Interpretation of $\varphi_{V}^{n}$ as the probability to be scattered $n$ times.}
\end{figure}

\section{Examples of Functional Measures}
\label{sec:examples}

\subsection{Diffusion Equation}

As a first example we consider the diffusion equation
\begin{equation}\label{EqDif}
    \frac{\partial f}{\partial t} = D\;\nabla^{2} f - V f,
\end{equation}
Where $D > 0$ is constant known as the diffusion coefficient. In this case the operator $A$ is given by
\begin{equation*}
    A(f) = D\;\nabla^{2} f.
\end{equation*}
By means of Fourier transform it is found that the fundamental solution $\varphi$ of this operator is
\begin{equation*}
    \varphi(x-y,t) = \frac{1}{\sqrt{4\pi Dt}} \mbox{\Large\(e^{-\frac{|x-y|^{2}}{4Dt}}\)}.
\end{equation*}
This function is positive, thus it restricts to the space of continuous paths. Therefore,  integrals with respect to the obtained measure $\mu_{\varphi}^{x,y;t}$ are computed using the functional integral formula, from which we obtain that
\begin{align*}
    &\int_{\Omega_{x,y;t}} f\left(\int_{0}^{t}V(\gamma(s))\;ds\right)\;d\mu_{\varphi}^{x,y;t}(\gamma)\\ = &\lim_{N\to\infty} \frac{1}{\sqrt{4\pi Dt}}\int \cdots\int f\left(\sum_{k=1}^{N}V(x_{k})\Delta t(N)\right) \mbox{\Large\(e^{-\sum_{i=1}^{N+1}\frac{|x_{i}-x_{i-1}|^{2}}{4D\Delta t(N)}}\)}\;dx_{1}\cdots dx_{N}.
\end{align*}
In particular, if $f(s) = e^{-s}$ we have that
\begin{equation*}
    \varphi_{V}(x-y,t) = \lim_{N\to\infty} \frac{1}{\sqrt{4\pi Dt}}\int \cdots\int \mbox{\Large\(e^{-\sum_{i=1}^{N+1}\left(V(x_{i})\Delta t(N) + \frac{|x_{i}-x_{i-1}|^{2}}{4D\Delta t(N)}\right)}\)}\;dx_{1}\cdots dx_{N}.
\end{equation*}

If $V$ is bounded from below then the function $\varphi_{V}$ is well defined and it is the fundamental solution of the original diffusion equation (\ref{EqDif}). The obtained measure $\mu_{\varphi}^{x,y;t}$ is known as the Wiener measure and denoted by $\mathcal{W}$.

\subsection{$n$-th Derivative Operator}

As a second example we consider the differential equation
\begin{equation}\label{EqDn}
    \frac{\partial f}{\partial t} = \frac{\partial^{n} f}{\partial x^{n}} - V f.
\end{equation}
In this case the operator $A$ is given by the $n$-th partial derivative,
\begin{equation*}
    A(f) = \frac{\partial^{n} f}{\partial x^{n}}.
\end{equation*}
By means of Fourier transform it is found that the fundamental solution $\varphi$ of this operator is
\begin{equation*}
    \varphi(x-y,t) = \frac{1}{2\pi} \int \mbox{\Large\(e^{ik(x-y) - t(ik)^{n}}\)}\;dk,
\end{equation*}
where the integral is defined as a distribution. It follows that integration of cylinder functions with respect to the functional measure $\mu_{\varphi}^{x,y;t}$ is given by
\begin{equation*}
    \int_{\Omega_{x,y;t}} F(\gamma)\;d\mu_{\varphi}^{x,y;t}(\gamma) = \frac{1}{(2\pi)^{N}}\int \int \tilde{F}(x_{1},\ldots,x_{N})\;\mbox{\Large\(e^{ik_{i}(x_{i}-x_{i-1}) - \Delta t(N) (ik_{i})^{n} }\)}\;dx_{i}\;dk_{i}.
\end{equation*}
Since the fundamental solution may not be real, let alone positive, we introduce the sequence of cylinder functions $e_{n}\colon \Omega_{x,y;t} \to \mathbb{R}$ given by
\begin{equation*}
    e_{V,n}(\gamma) = \mbox{\Large\(e^{\sum_{k=1}^{n}V(\gamma(t_{k}))\Delta t(n)}\)},
\end{equation*}
where $t_{k} = k\frac{t}{n}$ and $\Delta t(n) = \frac{t}{n}$. The fundamental solution of (\ref{EqDn}) in terms of $\mu_{\varphi}^{x,y;t}$ is then
\begin{equation*}
    \varphi_{V}(x-y,t) = \lim_{n\to\infty} \int_{\Omega_{x,y;t}} e_{V,n}(\gamma)\;d\mu_{\varphi}^{x,y;t}(\gamma).
\end{equation*}
Once again, this limit may not exist for different values of $x$, $y$ and $t$, however, it does exist except on a set of measure zero and defines an element of $L^{2}$.

\subsection{Fokker-Planck Equation}

As a third example we consider the Fokker-Planck equation
\begin{equation*}
    \frac{\partial f}{\partial t} = \frac{1}{2}\sum_{i,j=1}^{n}Q_{ij}(x) \frac{\partial^{2}f}{\partial x_{i} \partial x_{j}} + \sum_{i=1}^{n}\rho_{i}(x) \frac{\partial f}{\partial x_{i}} + V(x) f,
\end{equation*}
where for every fixed $x$ the matrix $(Q_{ij})_{i,j\leq n}$ is symmetric and positive definite. In this case the operator $A$ is given by
\begin{equation*}
    A(f) = \frac{1}{2}\sum_{i,j=1}^{n}Q_{ij}(x) \frac{\partial^{2}f}{\partial x_{i} \partial x_{j}} + \sum_{i=1}^{n}\rho_{i}(x) \frac{\partial f}{\partial x_{i}}.
\end{equation*}
Under very mild regularity conditions on $Q$ and $\rho$ it can be proven that $A$ generates a $C_{0}$-semigroup. If $A$ also admits a fundamental solution $\varphi$ then there exists a functional measure $\mu_{\varphi}^{x,y;t}$ such that the fundamental solution of the Fokker-Planck equation is given by $\varphi_{V}$. We cannot give an explicit form for $\varphi_{V}$ unless we compute $\varphi$ which can only be done if the functions $Q$ and $\rho$ are given explicitly. However, it is enough to show that our construction generalizes the MSRJD integral. In one spatial dimension the Fokker-Planck equation can be written as
\begin{equation*}
    \frac{\partial f}{\partial t} = \frac{1}{2}Q(x) \frac{\partial^{2}f}{\partial x^{2}} + \rho(x) \frac{\partial}{\partial x}\left(\rho(x)  f\right) + V(x) f.
\end{equation*}
In this case the MSRJD integral is usually written in the notation of physics as
\begin{equation*}
    \varphi_{V} = \int \mbox{\Large\(e^{-\int L\;dt}\)}\;\mathcal{D}\gamma
\end{equation*}
where
\begin{equation*}
    L = \frac{1}{2 Q(\gamma)} \dot{\gamma}^{2} + \dot{\gamma}\frac{\rho(\gamma)}{Q(\gamma)} + \frac{\rho(\gamma)^{2}}{2 Q(\gamma)}
\end{equation*}
and the symbol $\mathcal{D}\gamma$ does not indicate an integral with respect to a measure, rather than some sort of summation must be done along all paths in a certain space. Our construction shows that this summation is the integral with respect to the measure $\mu_{\varphi}^{x,y;t}$.

\subsection{General Lagrangian}

In physics, the study of functional integration consists on associating to any Lagrangian $L(\gamma)$ an expression of the form $\int e^{-\int L\;dt}\mathcal{D}\gamma$. This is usually done by discretizing the variables of the Lagrangian and computing a limit over an arbitrary number of integrals, similar to the functional integral formula. The simplest Lagrangian for particles is of the form
\begin{equation*}
    L(\gamma,\dot{\gamma}) = \frac{1}{2}\dot{\gamma}^{2} - V(\gamma),
\end{equation*}
which is discretized using the following substitutions
\begin{align*}
    \gamma &\longmapsto x_{i},\\
    \dot{\gamma}_{i} &\longmapsto \frac{x_{i} - x_{i-1}}{t_{i} - t_{i-1}}.
\end{align*}
It is then easy to see that each part of the Lagrangian corresponds to a fundamental solution of an operator.

This  discretization method does not work for more complex Lagrangians. The problem lies on finding the proper discretization for functions of derivatives. A partial solution for this may be as follows: to terms in the Lagrangian of the form $\dot{\gamma}_{i}^{n}$ associate the operator $\partial_{i}^{n}$. Extend linearly to find the associated operator to a Lagrangian containing summands of this form, excluding the potential term. The correct discretization is then given by the fundamental solution of this operator.

\section{Extension to Paths in Infinite Dimensional Spaces}
\label{sec:inf}

\subsection{Transition Functions and Measures}

In the construction of functional measures in $\hat{\mathbb{R}}^{n}$ we only used a few general properties of this space. Because of this, we will extend the previous results to more general spaces. In place of $\hat{\mathbb{R}}^{n}$, we will assume that $(X,\tau)$ is a compact Hausdorff topological group and will denote the group operation as sum. Define
\begin{equation*}
    \Omega_{t} = \prod_{i\in [0,t]}X
\end{equation*}
with its product topology $\tau_{\prod}$. It follows from Tychonoff's theorem that $(\Omega,\tau_{\prod})$ is a compact Haussdorf space. For fixed $x,y\in X$ we define
\begin{equation*}
    \Omega_{x,y;t} = \{\gamma\in\Omega_{t}\;|\;\gamma(0) = x,\;\gamma(t) = y\}.
\end{equation*}
By a simple net argument one can verify that $\Omega_{x,y;t}$ is closed, thus also compact and Hausdorff in the subspace topology, which we will also denote by $\tau_{\prod}$. Let us consider the Borel $\sigma$-algebra $\sigma\left(\tau_{\prod}\right)$ of this space. We will work on the measurable space $\left(\Omega_{x,y;t},\sigma\left(\tau_{\prod}\right)\right)$. Assume that $\nu$ is a given measure in this measurable space. We say that a collection of integrable functions $\{\varphi_{s}(x) = \varphi(x,s)\;|\;s\in [0,t]\}$ is a $\nu$-\textbf{transition function} if it satisfies the compatibility condition
\begin{equation*}
    \int\varphi(x-y,v-s)\varphi(y-z,s-u)\;d\nu(y) = \varphi(x-z,v-u).
\end{equation*}
Every transition function defines a measure $\mu_{\varphi}^{x,y;t}$ in exactly the same way as before. The proof of the existence of functional measures is formally identical to the proof in $\hat{\mathbb{R}}^{n}$ since it does not rely on any properties of the space other than it being compact and Hausdorff. The full measure of continuous paths is similar but more complicated.

\begin{theorem}[Full Measure of Continuous Paths]
If $\varphi$ is a positive transition function on $(X,\tau)$ and
\begin{equation}\label{EqFullMeasureCondition2}
    \lim_{\delta\to 0}\int_{X\setminus V}\varphi(x_{2} - x_{1},\delta)\varphi(x_{1} - x,s)\;dx_{1} = 0
\end{equation}
for each $s\in [0,t]$, $x_{2}\in X$ and $V$ is a neighbourhood of $x_{2}$, then the set of elements in $\Omega_{x,y;t}$ that are continuous at every point has full measure respect to $\mu_{\varphi}^{x,y;t}$.
\end{theorem}
\begin{proof}
Define $D$ to be the set of paths in $\Omega_{x,y;t}$ that are discontinuous in at least one point and for $V\in \tau$ and $\delta > 0$ define $\omega(V,\delta)$ to be the set of paths in $\Omega_{x,y;t}$ for which there exist $r,s\in [0,t]$ such that $|r-s| < \delta$ and $\gamma(r)\in V$ but $\gamma(s)\notin V$. It follows that
\begin{equation*}
    D = \bigcup_{V\in \tau}\bigcap_{\delta > 0} \omega(V,\delta).
\end{equation*}
The intersection can be made countable by taking $\delta$ to be the reciprocal of a natural. To deal with the union we will use inner regularity. Using inner regularity and continuity of measure we have that
\begin{equation*}
    \mu_{\varphi}^{x,y;t}(D) = \sup_{V\in\tau}\;\lim_{\delta\to 0} \mu_{\varphi}^{x,y;t}(\omega(V,\delta)).
\end{equation*}
Then, it is enough to show that $\lim_{\delta\to 0} \mu_{\varphi}^{x,y;t}(\omega(\epsilon,\delta)) = 0$ for any $V\in\tau$. The rest of the proof follows in a similar manner as the previous case. For any finite and increasingly ordered subset $F = \{t_{i}\}_{i=1}^{N}$ such that $|t_{i} - t_{i-1}| < \delta$ we define $\omega(V;F)$ to be the set of elements $\gamma\in\Omega_{x,y;t}$ for which there exists $i\leq N$ such that $\gamma(t_{i})\in V$ but $\gamma(t_{i-1}) \notin V$. Denote by $\mathcal{C}_{\delta}$ the class of all such sets. Then
\begin{equation*}
    \omega(V,\delta) = \bigcup_{\omega\in\mathcal{C}_{\delta}}\omega(V;F).
\end{equation*}
Inner regularity of $\mu_{\varphi}^{x,y;t}$ implies that
\begin{equation*}
    \mu_{\varphi}^{x,y;t}(\omega(V,\delta)) = \sup_{\omega\in\mathcal{C}_{\delta}} \mu_{\varphi}^{x,y;t}(\omega(V;F)).
\end{equation*}
This shows that it is enough to verify that $\mu_{\varphi}^{x,y;t}(\omega(\epsilon,\delta;r,s))$ vanishes as $\delta \to 0$. The cumbersome computation of the uniform bound for the measure of these sets is done in exactly the same way as before, with $B = \{(x_{1},x_{2})\in X\times X\;|\;x_{1}\in V,\;x_{2}\notin V\}$. We will not repeat it here.
\end{proof}

The proof we gave for the functional integral formula also holds in this case.

\subsection{Distributions in Normed Spaces}

In the finite dimensional case, transition functions arise as fundamental solutions of operators in $L^{2}(\hat{\mathbb{R}}^{n})$. This relies on properties of distributions, which require special properties of $\mathbb{R}^{n}$. Many topological spaces are inadequate for generalizing distributions, since they lack a notion of differentiability. Smooth finite dimensional manifolds, however, admit reasonable generalizations. Infinite dimensional normed spaces have a suitable differentiability notion, but face the problem of not being locally compact. The same space with its weak topology is not locally compact but admits an exhaustion by compact sets if the space is reflexive. In doing this, one faces the additional difficulty of having problems with the domains of smooth functions. In the following generalization of distributions to normed spaces we show that these difficulties can be dealt with by carefully using both spaces.

Assume $(X,|\cdot|)$ is a reflexive normed space. If $\tau_{\omega}$ is its weak topology and $B_{X}$ its closed unit ball, then $(B_{X},\tau_{\omega})$ is a compact  Hausdorff topological space. This space is not a topological group, since the sum of two elements may not be in $B_{X}$. However, the previous results apply if we restrict the sum to be defined only if the sum of two elements remains in $B_{X}$ (it is still not a group). The advantage of this setting is that differentiability is well defined, either in the sense of Frechet or Gateaux. We define $f\colon X \to \mathbb{R}$ to be \textbf{continuously Gateaux differentiable} of order $k$ if it is Gateaux differentiable of order $k$ and for every $l\leq k$ its $l$-th Gateaux derivative $D f\colon X \to \mathbb{R}^{X^{l}}$ is weakly continuous and $l$-linear at each point. For $K\subset X$ define $C_{c}^{k}(K)$ to be the set of continuously Gateaux differentiable functions with support in $K$ and
\begin{equation*}
    C_{c}^{\infty}(K) = \bigcap_{k\in\mathbb{N}}C_{c}^{k}(K).
\end{equation*}
The set $C_{c}^{\infty}(B_{X})$ is not properly defined since it has empty interior in $(X,\tau_{\omega})$, however, its interior is non-empty in $(X,|\cdot|)$, which shows that it is properly defined as a set and so are all of the Gateaux derivatives of its elements. For $n\in\mathbb{N}$ equip each $C_{c}^{\infty}(nB_{X})$ with the countable collection of seminorms
\begin{equation*}
    p_{n,k}(f) = \sup_{x\in nB_{X}} |D_{x}^{k}f|.
\end{equation*}
For each $n\in\mathbb{N}$ we will denote the topology generated by these seminorms as $\tau_{n}$. Since each collection is countable and $p_{n,0}$ is actually a norm it follows that each collection is separating and thus 
$\tau_{n}$ is given by a translation invariant metric. Since $f_{m} \xrightarrow[]{\tau_{n}} f$ if and only if $p_{n,k}(f_{m} \to f) \to 0$ for each $k\in\mathbb{N}$, it follows that each $C_{c}^{\infty}(nB_{X})$ is actually a Frechet space. Using the Arzela-Ascoli and the mean value theorems, it can be shown that each of these spaces has the Heine-Borel property. Notice that a sequence $(f_{m})_{m\in\mathbb{N}}$ in $C_{c}^{\infty}(X)$ satisfies $f_{m} \xrightarrow[]{\tau_{n}} f$ if and only if the sequence and $f$ belong to this space and $p_{n,k}(f_{n} - f) \to 0$ for each $k\in\mathbb{N}$. It follows that for $n \leq l$ the inclusions $C_{c}^{\infty}(nB_{X}) \hookrightarrow  C_{c}^{\infty}(lB_{X})$ are continuous and $C_{c}^{\infty}(nB_{X}) \cap \tau_{l} = \tau_{n}$.

Our previous considerations show that the hypotheses of the Diudonné-Schwarz theorem are satisfied. From this we obtain many properties, which we list using the notation $D(X) = C_{c}^{\infty}(X)$.

\begin{enumerate}
    \item There exists a locally convex Hausdorff topology $\tau_{D}$ in $D(X)$ such that $\tau_{D} \cap C_{c}^{\infty}(nB_{X}) = \tau_{n}$.

    \item $A \subset D(X)$ is $\tau_{D}$-bounded if and only if there exists an $n\in\mathbb{N}$ such that $B \subset C_{c}^{\infty}(nB_{X})$ and $B$ is $\tau_{n}$-bounded.

    \item A sequence $(f_{m})_{m\in\mathbb{N}}$ in $D(X)$ converges to $f$ with respect to $\tau_{D}$ if and only if the union of the supports of every $f_{m}$ is compact and $D^{k}f_{m} \xrightarrow[]{unif} D^{k}f$ on this set for each $k\in\mathbb{N}$.

    \item A linear functional $\Lambda\colon (D(X),\tau_{D}) \to \mathbb{R}$ is continuous if and only if it is sequentially continuous.
\end{enumerate}

The last two propierties show that $D'(X) = (D(X),\tau_{D})^{\ast}$ is a suitable candidate to generalize the theory of distributions on the reflexive space $(X,|\cdot|)$. This space is of course equipped with its weak-$\ast$ topology, that is, the topology of pointwise convergence. The space $D'(X)$ clearly contains the set of Borel measures of $(X,\tau_{\omega})$, in particular, the Dirac $\delta$ distributions are defined.

From this we can consider a Borel measure $\nu$ in $(X,\tau_{\omega})$ and an operator $A\in B(L^{2}(X,\sigma(\tau_{\omega}),\nu))$. It makes sense to define a fundamental solution to the abstract Cauchy problem
\begin{equation*}
    \partial_{t} f = A(f)
\end{equation*}
in the same way we did before. Furthermore, our previous results imply that if the operator $A$ generates a $C_{0}$-semigroup and admits a fundamental solution then this fundamental solution is a transition function.

This method, however, has several shortcomings. First, we do not know how easy it is for an operator in an $L^{2}(X)$ to admit a fundamental solution. Second, when defining the spaces $C_{c}^{\infty}(K)$ we imposed that every Gateaux derivative must be weakly continuous, which may be too strong to admit a reasonable amount of functions.

\section{Concluding remarks}
\label{sec:con}

In this work, we generalize the notion of Wiener measure by associating it with certain differential operators. The main idea is simple and is based upon the use of the relationship that exists between fundamental solutions, transition functions, and functional measures. 

First, we present an exact definition of the fundamental solution of an operator in terms of a functional integral. Then, we show that for certain operators the fundamental solution is also a transition function in the sense of Kolmogorov's compatibility condition. Then, we show that in general, any transition function defines a functional measure on a space of rough paths. These two facts induce the relationship between operators and functional measures defined, however, on rough paths.

The next step consists in showing that there exists a set of continuous paths that is also compatible with the functional measure of the fundamental solution. To show this, it is necessary to assume that the transition function satisfies a rapid decay condition, satisfied by every fundamental solution. We then use a series development to compute general functional integrals, where each of the terms of the series can be easily calculated, giving rise to the possibility of computing it explicitly or approximately. The interesting byproduct of this representation of the functional integral is that said integral is actually the solution of an abstract Cauchy problem.

We used the above results to derive explicit solutions in terms of functional integrals for well-known equations, such as the diffusion and the Fokker-Planck equations. We also analyzed the case of the nth-derivative operator and studied the one-dimensional general Lagrangian of classical mechanics.

We also outlined some thoughts about the possibility of generalizing the above result to include functional measures with continuous paths in infinite dimensional spaces. Preliminary results show that the analysis can be carried out by considering infinite-dimensional distributions, but some problems appear when trying to define operators that allow the existence of fundamental solutions. This is a task for future investigations.

\section*{Acknowledgements}
This work was supported by 
UNAM PASPA–DGAPA.\\

\noindent \textbf{Data Availability} Data sharing is not applicable to this article as no datasets were generated or analyzed
during the study.\\

\noindent
\textbf{Declarations}\\

\noindent
\textbf{Conflict of interest} The authors declare that they have no conflict of interest.

\end{document}